\documentclass[10pt,twoside]{amsart}
\usepackage{enumerate}
\usepackage{graphicx}
\usepackage{amsfonts}
\usepackage{graphicx}
\usepackage{amsmath}
\usepackage{amsxtra}
\usepackage{amssymb}
\usepackage[latin1]{inputenc}
\usepackage{dsfont}

\def\YEAR{\year}\newcount\VOL\VOL=\YEAR\advance\VOL by-1995
\def\firstpage{235}\def\lastpage{257}
\def\received{November 18, 2008}\def\revised{March 11, 2009}
\def\communicated{Heinz Siedentop}

\makeatletter
\def\magnification{\afterassignment\m@g\count@}
\def\m@g{\mag=\count@\hsize6.5truein\vsize8.9truein\dimen\footins8truein}
\makeatother

\font\eightrm=cmr8
\font\caps=cmcsc10                    
\font\Caps=cmcsc10 scaled \magstep1   
\font\scaps=cmcsc8

\makeatletter
\@specialpagefalse\headheight=8.5pt
\def\DocMath{}
\def\DocMath{{\def\th{\thinspace}\scaps Documenta Math.}}
\renewcommand{\@oddfoot}{\hfill\scaps Documenta Mathematica 
    \number\VOL\  (\number\YEAR) \number\firstpage--\lastpage\hfill}
\renewcommand{\@evenfoot}{\ifnum\thepage>\lastpage\hfill\scaps
    Documenta Mathematica \number\VOL\  (\number\YEAR)\hfill\else\@oddfoot\fi}%

\renewcommand{\@evenhead}{%
    \ifnum\thepage>\lastpage\rlap{\thepage}\hfill%
    \else\rlap{\thepage}\slshape\leftmark\hfill{\caps\SAuthor}\hfill\fi}%
\renewcommand{\@oddhead}{%
    \ifnum\thepage=\firstpage{\DocMath\hfill\llap{\thepage}}%
    \else{\slshape\rightmark}\hfill{\caps\STitle}\hfill\llap{\thepage}\fi}%
\makeatother

\def\TSkip{\bigskip}
\newbox\TheTitle{\obeylines\gdef\GetTitle #1
\ShortTitle  #2
\SubTitle    #3
\Author      #4
\ShortAuthor #5
\EndTitle
{\setbox\TheTitle=\vbox{\baselineskip=20pt\let\par=\cr\obeylines%
\halign{\centerline{\Caps##}\cr\noalign{\medskip}\cr#1\cr}}%
	\copy\TheTitle\TSkip\TSkip%
\def\next{#2}\ifx\next\empty\gdef\STitle{#1}\else\gdef\STitle{#2}\fi%
\def\next{#3}\ifx\next\empty%
    \else\setbox\TheTitle=\vbox{\baselineskip=20pt\let\par=\cr\obeylines%
    \halign{\centerline{\caps##} #3\cr}}\copy\TheTitle\TSkip\TSkip\fi%
\centerline{\caps #4}\TSkip\TSkip%
\def\next{#5}\ifx\next\empty\gdef\SAuthor{#4}\else\gdef\SAuthor{#5}\fi%
\ifx\received\empty\relax
    \else\centerline{\eightrm Received: \received}\fi%
\ifx\revised\empty\TSkip%
    \else\centerline{\eightrm Revised: \revised}\TSkip\fi%
\ifx\communicated\empty\relax
    \else\centerline{\eightrm Communicated by \communicated}\fi\TSkip\TSkip%
\catcode'015=5}}\def\Title{\obeylines\GetTitle}
\def\Abstract{\begingroup\narrower
    \parskip=\medskipamount\parindent=0pt{\caps Abstract. }}
\def\EndAbstract{\par\endgroup\TSkip}

\long\def\MSC#1\EndMSC{\def\arg{#1}\ifx\arg\empty\relax\else
     {\par\narrower\noindent%
     2000 Mathematics Subject Classification: #1\par}\fi}

\long\def\KEY#1\EndKEY{\def\arg{#1}\ifx\arg\empty\relax\else
	{\par\narrower\noindent Keywords and Phrases: #1\par}\fi\TSkip}

\newbox\TheAdd\def\Addresses{\vfill\copy\TheAdd\vfill
    \ifodd\number\lastpage\vfill\eject\phantom{.}\vfill\eject\fi}
{\obeylines\gdef\GetAddress #1
\Address #2
\Address #3
\Address #4
\EndAddress
{\def\xs{6.0truecm}\parindent=0pt
\setbox0=\vtop{{\obeylines\hsize=\xs#1\par}}\def\next{#2}
\ifx\next\empty 
     \setbox\TheAdd=\hbox to\hsize{\hfill\copy0\hfill}
\else\setbox1=\vtop{{\obeylines\hsize=\xs#2\par}}\def\next{#3}
\ifx\next\empty 
     \setbox\TheAdd=\hbox to\hsize{\hfill\copy0\hfill\copy1\hfill}
\else\setbox2=\vtop{{\obeylines\hsize=\xs#3\par}}\def\next{#4}
\ifx\next\empty\ 
     \setbox\TheAdd=\vtop{\hbox to\hsize{\hfill\copy0\hfill\copy1\hfill}
                \vskip20pt\hbox to\hsize{\hfill\copy2\hfill}}
\else\setbox3=\vtop{{\obeylines\hsize=\xs#4\par}}
     \setbox\TheAdd=\vtop{\hbox to\hsize{\hfill\copy0\hfill\copy1\hfill}
	        \vskip20pt\hbox to\hsize{\hfill\copy2\hfill\copy3\hfill}}
\fi\fi\fi\catcode'015=5}}\gdef\Address{\obeylines\GetAddress}

\begin{document}
\Title
The Allegretto-Piepenbrink Theorem 
for Strongly Local Dirichlet Forms
\ShortTitle
The Allegretto-Piepenbrink Theorem
\SubTitle
\textit{Dedicated to J\"urgen Voigt in celebration of his 65th birthday}

\Author
Daniel Lenz, Peter Stollmann, Ivan Veseli\'c
\ShortAuthor
D. Lenz, P. Stollmann, I. Veseli\'c
\EndTitle

\Abstract
The existence of positive weak solutions is related to spectral information on the corresponding
partial differential operator.

\EndAbstract
\MSC
35P05, 81Q10

\EndMSC
\KEY
\EndKEY
\Address
Daniel Lenz
Mathematisches Institut
Friedrich-Schiller 
Universit\"at Jena
Ernst-Abb\'{e} Platz 2
07743 Jena
Germany
{\small
daniel.lenz@uni-jena.de
http://www.analysis-lenz.uni-jena.de/}

\Address
Peter Stollmann,
Fakult\"at f\"ur Mathematik
           TU Chemnitz
           09107 Chemnitz
           Germany
{\small
 stollman@mathematik.tu-chemnitz.de}

\Address
Ivan Veseli\'{c}
Emmy-Noether-Programme 
of the DFG\, \& 
Fakult\"at f\"ur Mathematik
TU Chemnitz 
09107\ Chemnitz
Germany
{\small
http://www.tu-chemnitz.de/mathematik/enp/}

\Address
\EndAddress
\overfullrule=5pt

\newcommand{\ZZ}{{\mathbb Z}}
\newcommand{\RR}{{\mathbb R}}
\newcommand{\R}{{\mathbb R}}

\newcommand{\CC}{{\mathbb C}}
\newcommand{\NN}{{\mathbb N}}
\newcommand{\N}{{\mathbb N}}

\newcommand{\TT}{{\mathbb T}}
\newcommand{\KK}{{\mathbb K}}
\newcommand{\supp}{\mathrm{supp}\,}
\providecommand{\C}[1]{\mathcal{#1}}
\newcommand{\aaa}{\C{E}}
\newcommand{\cE}{\C{E}}
\newcommand{\DD}{\C{D}}
\newcommand{\cD}{\C{D}}
\newcommand{\capp}{\mathrm{cap}}
\newcommand{\dist}{\mathrm{dist}}
\newcommand{\vol}{\mathrm{vol}}
\newcommand{\phim}{\Phi^{-1}}

\sloppy

\newtheorem{theorem}{Theorem}[section]
\newtheorem{acknowledgement}[theorem]{Acknowledgement}
\newtheorem{algorithm}[theorem]{Algorithm}
\newtheorem{axiom}[theorem]{Axiom}
\newtheorem{case}[theorem]{Case}
\newtheorem{claim}[theorem]{Claim}
\newtheorem{conclusion}[theorem]{Conclusion}
\newtheorem{condition}[theorem]{Condition}
\newtheorem{conjecture}[theorem]{Conjecture}
\newtheorem{corollary}[theorem]{Corollary}
\newtheorem{criterion}[theorem]{Criterion}
\newtheorem{definition}[theorem]{Definition}
\newtheorem{example}[theorem]{Example}
\newtheorem{exercise}[theorem]{Exercise}
\newtheorem{lemma}[theorem]{Lemma}
\newtheorem{notation}[theorem]{Notation}
\newtheorem{problem}[theorem]{Problem}
\newtheorem{proposition}[theorem]{Proposition}
\newtheorem{remark}[theorem]{Remark}
\newtheorem{solution}[theorem]{Solution}
\newtheorem{summary}[theorem]{Summary}

\section*{Introduction}
The Allegretto-Piepenbrink theorem relates solutions and spectra of 2nd order partial differential operators $H$ and has quite some history,
cf.~ \cite{Agmon-83,Allegretto-74, Allegretto-79, Allegretto-81,MossP-78,Piepenbrink-74,Piepenbrink-77, Simon-82, Pinchover-07,Pinsky-95}.

One way to phrase it is that the supremum of those real $E$   for which a nontrivial positive solution of $H\Phi=E\Phi$ exists coincides with the infimum of the spectrum of $H$. In noncompact cases this can be sharpened in the sense that nontrivial positive solutions of the above equation exist for all $E\le\inf\sigma(H)$.

In the present paper we consider the Allegretto-Piepenbrink theorem in a
general setting in the sense that the coefficients that are allowed may be
very singular. In fact, we regard $H=H_0+\nu$, where $H_0$ is the generator of
a strongly local Dirichlet form and $\nu$ is a suitable measure perturbation.
Let us stress, however, that one main motivation for the present work is the
conceptual simplicity that goes along with the generalisation.

The Allegretto-Piepenbrink theorem as stated above consists of two statements:
the first one is the fact that positive solutions can only exist for $E$ below
the spectrum. Turned around this means that the existence of a nontrivial
positive solution of $H\Phi=E\Phi$ implies that $H\ge E$. For a strong enough
notion of positivity, this comes from a ``ground state transformation''. We
present this simple extension of known classical results in Section 2, after
introducing the necessary set-up in Section 1. For the ground state
transformation not much structure is needed.

For the converse statement, the existence of positive solutions below
$\sigma(H)$, we need more properties of $H$ and the underlying space:
noncompactness, irreducibility and what we call a Harnack principle. All these
analytic properties are well established in the classical case. Given these
tools, we prove this part of the Allegretto-Piepenbrink theorem in Section 3
with arguments reminiscent of the corresponding discussion in
\cite{Simon-82,CyconFKS-87}. For somewhat complementary results we refer to
\cite{BoutetdeMonvelLS-08} where it is shown that existence of a nontrivial
subexponentially bounded solution of $H\Phi=E\Phi$ yields that
$E\in\sigma(H)$. This implies, in particular, that the positive solutions we
construct for energies below the spectrum cannot behave to well near infinity.
We dedicate this paper to J\"urgen Voigt - teacher, collaborator and friend -
in deep gratitude and wish him many more years of fun in analysis.

\section{Basics and notation concerning strongly local Dirichlet forms and measure perturbations}
\subsection*{Dirichlet forms}
We will now describe the set-up;
we refer to \cite{Fukushima-80} as the
classical standard reference as well as
\cite{BouleauH-91,Davies-90b,FukushimaOT-94,MaR-92} for literature on
Dirichlet forms. Let us emphasize that in contrast to most of the work done on
Dirichlet forms, we treat real and complex function
spaces at the same time and write $\KK$ to denote either $\RR$ or $\CC$.

Throughout we will work with a locally compact, separable metric space
$X$ endowed with a positive Radon measure $m$ with $ \supp m=X$.

The central object of our studies is a regular
Dirichlet form $\mathcal{E}$ with domain $\mathcal{D}$ in $L^2(X)$ and
the selfadjoint operator $H_0$ associated with $\mathcal{E}$.
Let us recall the basic terminology of Dirichlet forms:
Consider a dense subspace $\mathcal{D} \subset L^2(X,m)$  and a
sesquilinear and non-negative map $\mathcal{E}\colon\mathcal{D} \times \mathcal{D} \rightarrow \KK$
such that $\mathcal{D}$ is closed with respect to the energy
norm $\|\,\cdot\,\|_\mathcal{E}$, given by
$$
\|u\|_\mathcal{E}^2=\mathcal{E}[u,u] +\| u\|_{L^2(X,m)}^2,
$$
in which case one speaks of a \textit{closed form} in $L^2(X,m)$. In the sequel we will write
$$\mathcal{E}[u]:= \mathcal{E} [u,u]. $$
  The selfadjoint operator $H_0$
associated with $\mathcal{E}$ is then characterized by
$$
D(H_0)\subset \mathcal{D} \ \mbox{and } \mathcal{E}[f,v]=(H_0f\mid
v)\quad (f\in D(H_0), v\in \mathcal{D}).
$$
Such a closed form is said to be a \textit{Dirichlet form} if
$\mathcal{D}$ is stable under certain pointwise operations; more
precisely, $T:\KK\to\KK$ is called a \emph{normal contraction} if
$T(0)=0$ and $|T(\xi)-T(\zeta)|\le |\xi -\zeta|$ for any
$\xi,\zeta\in\KK$ and we require that for any $u\in \mathcal{D}$ also
$$
T\circ u\in \mathcal{D}\mbox{ and }\mathcal{E}[T\circ u]\le
\mathcal{E}[u].
$$
Here we used the original condition from \cite{BeurlingD-59} that
applies in the real and the complex case at the same time. Today,
particularly in the real case, it is mostly expressed in an equivalent
but formally weaker statement involving $u\vee 0$ and $u\wedge 1$, see
\cite{Fukushima-80}, Thm. 1.4.1 and  \cite{MaR-92}, Section I.4.

A Dirichlet form is called \textit{regular} if $\mathcal{D} \cap
C_c(X)$ is large enough so that it is dense both in $(\mathcal{D},\| \cdot \|_\mathcal{E})$ and
$(C_c(X),\| \cdot \|_{\infty })$, where $C_c(X)$ denotes the space of
continuous functions with compact support.
\subsection*{Capacity}
Due to regularity, we find a set function, the \emph{capacity} that allows to measure the size of sets in a way that is adapted to the form $\aaa$:
For $U\subset X$, $U$ open,
$$
\capp(U):=\inf\{\| v \|^2_\mathcal{E} \mid v\in\DD, \chi_U\le v\} , (\inf\emptyset =\infty ),
$$
and
$$
\capp(A):=\inf\{ \capp(U) \mid A\subset U\}
$$
(see \cite{Fukushima-80}, p. 61f.).  We say that a property holds
\emph{quasi-everywhere}, short \emph{q.e.}, if it holds outside a set of
capacity $0$. A function $f:X\to\KK$ is said to be \emph{quasi-continuous},
\emph{q.c.} for short, if, for any $\varepsilon >0$ there is an open set
$U\subset X$ with $\capp(U)\le \varepsilon$ so that the restriction of $f$ to
$X\setminus U$ is continuous.

A fundamental result in the theory of Dirichlet forms says that every
$u\in\DD$ admits a q.c. representative $\tilde{u}\in u$ (recall that $u\in
L^2(X,m)$ is an equivalence class of functions) and that two such q.c.
representatives agree q.e. Moreover, for every Cauchy sequence $(u_n)$ in
$(\mathcal{D},\| \cdot \|_\mathcal{E})$ there is a subsequence $(u_{n_k})$
such that the $(\tilde{u}_{n_k})$ converge q.e. (see \cite{Fukushima-80},
p.64f).
\subsection*{Measure perturbations}

We will be dealing with Schr\"odinger type operators, i.e.,
perturbations $H=H_0+V$ for suitable potentials $V$. In fact, we can
even include measures as potentials. Here, we follow the approach from   \cite{Stollmann-92,StollmannV-96}. Measure perturbations have been regarded
by a number of authors in different contexts, see e.g.
\cite{AlbeverioM-91,Hansen-99a,Sturm-94b} and the references there.

We denote by
$\mathcal{M}_{R}(U)$ the signed Radon measures on the open subset $U$ of $X$
and by $\mathcal{M}_{R,0}(U)$ the subset of measures
$\nu$ that do not charge sets of capacity
$0$, i.e., those measures with $\nu(B)=0$ for every Borel set $B$ with
$\capp(B)=0$. In case that $\nu=\nu_+-\nu_-\in \mathcal{M}_{R,0}(X)$ we can define
$$
\nu[u,v]=\int_X\tilde{u}\overline{\tilde{v}}d\nu\mbox{  for  }u,v\in\DD
\mbox{  with  }\tilde{u},\tilde{v}\in L^2(X,\nu_++\nu_-) .
$$
We
 have to rely
upon more restrictive assumptions concerning the negative part $\nu_-$
of our measure perturbation. We write $\mathcal{M}_{R,1}$ for those
measures $\nu\in \mathcal{M}_{R}(X)$ that are $\aaa$-bounded with bound less than one; i.e.
measures $\nu$ for which there is a $\kappa<1$ and a $c_\kappa$ such that
$$
\nu[u,u]\le \kappa\aaa[u] + c_\kappa\| u\|^2 .
$$
The set $\mathcal{M}_{R,1}$ can easily be seen to be a subset of
$\mathcal{M}_{R,0}$. We write $\nu \in \mathcal{M}_{R,0}- \mathcal{M}_{R,1}$
if the positive part $\nu_+$ of the measure is in $\mathcal{M}_{R,0}$ and the negative
$\nu_-$ is in $\mathcal{M}_{R,1}$.

By the KLMN theorem (see \cite{ReedS-75}, p. 167), the sum
$\aaa+\nu$ given  by $D(\aaa+\nu)=\{u\in\DD\mid \tilde{u}\in L^2(X,\nu_+)\}$
is closed and densely defined (in fact $\mathcal{D} \cap
C_c(X)\subset D(\aaa+\nu)$) for
$\nu \in \mathcal{M}_{R,0}- \mathcal{M}_{R,1}$.
We denote the associated selfadjoint operator by $H_0+\nu$.
An important special case is given by $\nu=Vdm$ with $V\in L^1_{\text{loc}}(X)$. As done in various papers, one can allow for more singular measures, a direction we are not going to explore here due to the technicalities involved.

\subsection*{Approximation and Regularity}
By assumption the Dirichlet form $(\cE, \cD)$ is regular.
We show now that this property carries over to the perturbed form
$(\cE+\nu,D(\cE+\nu))$. Along the way we prove an approximation
result which will be useful in the context of Theorem  \ref{t:transformation}.
It will be convenient to introduce a notation for the natural norm
in $D(\cE+\nu)$. For all $\psi \in D(\cE+\nu)$ we define
\[
 \|\psi \|_{\cE+\nu}^2  := \|\psi \|_{\cE}^2 + \nu_+(\psi,\psi) \, .
\]

\begin{lemma} \label{l:approximation}
Let $\nu\in\ \mathcal{M}_{R,0}-\mathcal{M}_{R,1}$, and $\cE$ and $\cE+ \nu$ be as above.
Then
\begin{enumerate}[(a)]
 \item
For each $u \in D(\cE+\nu)$ there exists a sequence $(u_n)$ in $\cD \cap L_c^\infty(X)$ such that
$|u_n | \le |u|$ for all $n\in \NN$ and $\|u-u_n\|_{\cE+ \nu}\to 0$ for $n \to \infty$.
\item
For any $ v \in \cD \cap L_c^\infty(X)$ with $v\geq 0$  and any  $\eta \in \cD\cap C_c (X)$ with $\eta \equiv 1$ on the support of $v$  there exists
a sequence $(\phi_n)$ in $\cD\cap C_c (X)$ with $\phi_n \to v$ in $(D(\cE+\nu),\|\cdot\|_{\cE+\nu})$
and $0\le v, \phi_n \le \eta$ for all $n\in \NN$.
\end{enumerate}
In particular, $\cD \cap C_c(X)$ is dense in $(D(\cE+\nu),\|\cdot\|_{\cE+ \nu})$
and the form $(\cE+\nu,D(\cE+\nu))$ is regular.
\end{lemma}
Note that $\cD \cap L_c^\infty(X)\subset D(\cE+ \nu)$.
\begin{proof}
By splitting $u$ into its real and imaginary and then positive and negative part
we can assume afterwards that $u\ge0$.

We now prove the first statement.
Since $\cE$ is regular there exists a sequence $(\phi_n)$ in $\cD \cap C_c(X)$ such that
$\|u-\phi_n\|_{\cE} \to 0$. By the contraction property of Dirichlet forms
we can suppose that $\phi_n\ge 0$ and deduce that $u_n:=\phi_n \wedge u  \to u$
in $(\cD,\|\cdot\|_\cE)$ as well. (Note that $u_n = T(\phi_n -u)$ with the normal contraction $T: \R\longrightarrow \R$, $T(y) =y$ for $y\leq 0$ and $T(y)=0$ for $y\geq 0$.)
Choosing a subsequence, if necessary, we can make sure that $\tilde u_n \to \tilde u$
q.e. Therefore $\tilde u_n \to \tilde u$ a.e. with respect to $\nu_+$ and $\nu_-$.
Now $(\cE+\nu)$-convergence follows by Lebesgue's dominated convergence theorem.

Now we turn to the proof of the second statement.
Without loss of generality we may chose $0\le v \le 1$.
Consider the convex set
\[
C:= \{ \phi \in \cD \cap C_c(X) \mid 0 \le \phi \le \eta\} \,
\]
Since $C$ is convex, its weak and norm closure in $(D(\cE+\nu),\|\cdot\|_{\cE+\nu})$
coincide.
Therefore it suffices to construct a sequence $(\phi_n)\subset C$
that is bounded w.r.t.{} $\|\cdot\|_{\cE+\nu}$ and  converges to $\tilde v$ q.e.
By regularity we can start with a sequence $(\psi_n)\subset \cD\cap C_c(X)$
such that $\psi_n \to v$ w.r.t $\|\cdot\|_{\cE}$ and
$\tilde \psi_n \to \tilde v$ q.e.
By the contraction property of Dirichlet forms
the sequence $\phi_n := 0\vee\psi_n \wedge \eta$ is bounded
in $(\cD, \|\cdot\|_{\cE})$. Since $0\le \phi_n \le \eta$,
$(\phi_n)$ is also bounded in $L^2(\nu_+ +\nu_-)$.
We finally prove the  'in particular' statemtent.  Since $\cE$ is regular, we can find  an
$\eta \in \cD\cap C_c(X), \, 0 \le \eta \le 1$
with $\eta \equiv 1$ on $\supp v$. Now,  the proof follows from the previous parts.
\end{proof}

\subsection*{Strong locality and the energy measure}
$\mathcal{E}$ is called \textit{strongly local} if
$$\mathcal{E}[u,v]=0$$
whenever $u$ is constant a.s. on the support of $v$.

The typical example one should keep in mind is the Laplacian
$$H_0=-\Delta \mbox{ on }L^2(\Omega ),\quad \Omega \subset \mathbb{R}^d\mbox{ open, }$$
in which case
$$
\mathcal{D}=W^{1,2}_0(\Omega )\mbox{ and
}\mathcal{E}[u,v]=\int_{\Omega }(\nabla u| \nabla v)dx.$$ Now we turn
to an important notion generalizing the measure $(\nabla u |\nabla
v)dx$ appearing above.

In fact, every strongly local, regular Dirichlet form $\mathcal{E}$
can be represented in the form
\[
\mathcal{E}[u,v] = \int_X d\Gamma (u,v)
\]
where $\Gamma $ is a nonnegative sesquilinear mapping from
$\mathcal{D}\times\mathcal{D}$ to the set of $\KK$-valued Radon
measures on $X$. It is determined by
\[
\int_X \phi\, d\Gamma(u,u) = \mathcal{E}[u,\phi u] -\frac12
\mathcal{E}[u^2,\phi ]
\]
for realvalued $u\in\DD$, $\phi\in \mathcal{D} \cap
C_c(X)$ and called \textit{energy measure}; see also \cite{BouleauH-91}.
We discuss properties of the energy measure next (see e.g. \cite{BouleauH-91, Fukushima-80, Sturm-94b}). The
energy measure satisfies the Leibniz rule,
\[
d\Gamma(u\cdot v,w)=ud\Gamma (v,w)+vd\Gamma (u,w),
\]
as well as the chain rule
\[
d\Gamma (\eta (u),w)=\eta'(u)d\Gamma (u,w).
\]
One can even insert functions from $\mathcal{D}_{\text{loc}}$ into
$d\Gamma$, where $\mathcal{D}_{\text{loc}}$ is the set
\[
\lbrace u\in L^2_{\text{loc}}\mid \text{for all compact }
K \subset X \text{ there is  }
\phi \in \mathcal{D} \text{ s.\,t. }
\phi =u  \ m\text{-a.e.{} on }K\},
\]
as is readily seen from the
following important property of the energy measure, \textbf{strong
  locality}:

Let $U$ be an open set in $X$ on which the function $\eta\in
\mathcal{D}_{\text{loc}}$ is constant, then
\label{local}
\begin{equation*}
  \chi_U d\Gamma (\eta,u) = 0,
\end{equation*}
for any $u\in \mathcal{D}$. This, in turn, is a consequence of the
strong locality of $\mathcal{E}$ and in fact equivalent to the
validity of the Leibniz rule.

We write $d\Gamma(u):=d\Gamma(u,u)$ and note that the energy measure
satisfies the \textbf{Cauchy-Schwarz inequality}:
\begin{eqnarray*}
  \int_X|fg|d|\Gamma(u,v)| & \le &
  \left(\int_X|f|^2d\Gamma(u)\right)^{\frac12}\left(\int_X|g|^2d\Gamma(v)\right)^{\frac12}\\
  & \le & \frac12 \int_X|f|^2d\Gamma(u)+ \frac12\int_X|g|^2d\Gamma(v) .
\end{eqnarray*}
In order to introduce weak solutions on open subsets of $X$, we extend $\aaa$
and $\nu[\cdot,\cdot]$ to $\mathcal{D}_{\text{loc}}(U)\times \mathcal{D}_c (U)
$: where,
$$
\mathcal{D}_{\text{loc}}(U):=
\lbrace u\in
L^2_{\text{loc}}(U)\mid \forall \text{compact}\;\:
K \subset U \exists\;
\phi \in \mathcal{D} \text{ s.\,t. }
\phi =u  \ m\text{-a.e.{} on }K\}
$$
$$
\mathcal{D}_c(U):=
\lbrace \varphi\in \mathcal{D}|\supp\varphi\mbox{  compact in }U\rbrace .
$$
For $u\in \mathcal{D}_{\text{loc}}(U),
\varphi\in \mathcal{D}_c(U)$ we define
$$
\aaa[u,\varphi]:=\aaa[\eta u,\varphi],
$$
where $\eta \in \mathcal{D}\cap C_c(U)$ is arbitrary with constant value
$1$ on the support of $\varphi$. This makes sense as the RHS does not depend
on the particular choice of $\eta$ by strong locality. In the same way, we can
extend $\nu[\cdot,\cdot]$, using that every $u\in \mathcal{D}_{\text{loc}}(U)$
admits a quasi continuous version $\tilde{u}$. Moreover, also $\Gamma$ extends
to a mapping $\Gamma:\mathcal{D}_{\text{loc}}(U)\times
\mathcal{D}_{\text{loc}}(U)\to \C{M}_R(U)$.

\smallskip

For completeness reasons we explicitly state the following lemma.
\begin{lemma}\label{localversion}
\begin{itemize}
 \item [(a) ] Let $\Psi\in \mathcal{D}_{\text{loc}}\cap L^\infty_ {\text{loc}}(X)$ and $\varphi \in \mathcal{D}\cap L^\infty_c (X)$ be given. Then, $\varphi \Psi$ belongs to $\mathcal{D}$.
 \item [(b) ] Let $\Psi\in \mathcal{D}_{\text{loc}}$ and $\varphi \in \mathcal{D}\cap L^\infty_c (X)$ be such that $d\Gamma(\varphi)\le C\cdot dm$. Then, $\varphi \Psi$ belongs to $\mathcal{D}$.
\end{itemize}
\end{lemma}
\begin{proof}  Let $K$ be the support of $\varphi$ and $V$ an open neighborhood of $K$.

(a) Locality and the Leibniz rule give
$$ \int d\Gamma (\varphi \Psi) = \int_K |\varphi|^2 d\Gamma (\Psi) + 2 \int_K \varphi \Psi d\Gamma (\varphi, \Psi) + \int_K |\Psi|^2 d\Gamma (\varphi).$$
Obviously, the first and the last term are finite and the middle one can be estimated by Cauchy Schwarz inequality. Putting this  together, we infer
 $\int d\Gamma (\varphi \Psi)<\infty$.

(b) Clearly, it suffices to treat the case $\Psi\ge 0$. Since $\Psi_n:=\Psi\wedge n$ is a normal contraction of $\Psi$ for every $n\in\NN$ it follows that $d\Gamma(\Psi_n)\le d\Gamma(\Psi)$. By part (a) we know that $\varphi \Psi_n\in \mathcal{D}$ and an estimate as above gives that
\begin{eqnarray*}
\cE (\varphi \Psi_n) &=& \int_X d\Gamma( \varphi \Psi_n)\\
&\le & 2 \left(\int_X \varphi^2d\Gamma(  \Psi_n)+ \int_X\Psi_n^2 d\Gamma( \varphi )\right)\\
&\le & 2 \left(\int_X \varphi^2d\Gamma(  \Psi)+ C\int_X \chi_V\Psi^2 dm\right) ,
\end{eqnarray*}
is bounded independently of $n\in\NN$. As $\varphi \Psi_n$ converge to $\varphi \Psi$ in $L^2 (X,m)$, an appeal to the Fatou type lemma for closed forms, \cite{MaR-92}, Lemma 2.12., p. 21 gives the assertion.
\end{proof}

We close this section by noting that both $\DD\cap C_c (X)$ and $\DD \cap L^\infty_c (X)$ are closed under multiplication  (due to Leibniz rule).

\subsection*{The intrinsic metric}
Using the energy measure one can define the \textit{intrinsic metric}
$\rho$ by
$$
\rho (x,y)=\sup \lbrace |u(x)-u(y)|\ | u\in
\mathcal{D}_{\text{loc}}\cap C(X) \mbox{ and } d\Gamma (u)\leq
dm\rbrace
$$
where the latter condition signifies that $\Gamma (u)$ is absolutely
continuous with respect to $m$ and the Radon-Nikodym derivative is
bounded by $1$ on $X$. Note that, in general, $\rho$ need not be a
metric. We say that $\aaa$ is \emph{strictly local} if $\rho$ is a metric that
induces the original topology on $X$. Note that this implies that  $X$ is connected, since
otherwise points in $x,y$ in different connected components would give $\rho(x,y)=\infty$, as characteristic functions of connected components are continuous
and have vanishing energy measure.
We denote the intrinsic balls by
$$B(x,r):=\{ y\in X| \rho(x,y)\le r\} .$$
An important consequence of the latter assumption is that the distance
function $\rho_x(\cdot):=\rho(x,\cdot)$ itself is a function in
$\mathcal{D}_{\text{loc}}$ with $d\Gamma(\rho_x)\le dm$, see
\cite{Sturm-94b}. This easily extends to the fact that for every
closed $E\subset X$ the function $\rho_E(x):= \inf\{ \rho(x,y)|y\in
E\}$ enjoys the same properties (see the Appendix of \cite{BoutetdeMonvelLS-08}). This has a very
important consequence. Whenever $\zeta : \RR\longrightarrow \RR$ is
continuously differentiable, and $\eta :=\zeta \circ \rho_E$, then
$\eta$ belongs to $\mathcal{D}_{\text{loc}}$ and satisfies
\begin{equation}\label{ac}
  d \Gamma (\eta) = (\zeta'\circ \rho_E)^2 d \Gamma (\rho_E)\leq (\zeta'\circ \rho_E)^2  dm.
\end{equation}
For this reason a lot of good cut-off functions are around in our context. More explicitly we note the following lemma (see \cite{BoutetdeMonvelLS-08} as well).

\begin{lemma} \label{cutoff} For any compact $K$ in $X$ there exists a $\varphi \in C_c (X)\cap \mathcal{D}$ with $\varphi \equiv 1$ on $K$, $\varphi \geq 0$ and $d\Gamma (\varphi)\leq C \, dm$ for some $C>0$. If $L$ is another compact set containing $K$ in its interior, then $\varphi$ can be chosen to have support in $L$.
\end{lemma}
\begin{proof} Let $r>0$ be the positive distance of $K$ to the complement of $L$. Choose a two times differentiable $\zeta : \R\to [0,\infty)$ with $\zeta (0)=1$ and support contained in $(-\infty,r)$. Then, $\zeta\circ \rho_K$ does the job by \eqref{ac}.
\end{proof}

\subsection*{Irreducibility}
We will now discuss a notion that will be crucial in the proof of the
existence of positive weak solutions below the spectrum. In what follows,
$\mathfrak{h}$ will denote a densely defined, closed semibounded form in
$L^2(X)$ with domain $D(\mathfrak{h})$ and positivity preserving semigroup
$(T_t;t\ge 0)$.  We denote by $H$ the associated operator.  Actually, the
cases of interest in this paper are $\mathfrak{h}=\aaa$ or
$\mathfrak{h}=\aaa+\nu$ with $\nu\in\ \mathcal{M}_{R,0}-\mathcal{M}_{R,1}$.
We refer to \cite{ReedS-78}, XIII.12 and a forthcoming paper \cite{LenzSV-pre}
for details. We say that $\mathfrak{h}$ is \emph{reducible}, if there is a
measurable set $M\subset X$ such that $M$ and its complement $M^c$ are
nontrivial (have positive measure) and $L^2(M)$ is a reducing subspace for
$M$, i.e., ${\mathds{1}}_MD(\mathfrak{h})\subset D(\mathfrak{h})$,
$\mathfrak{h}$ restricted to ${\mathds{1}}_MD(\mathfrak{h})$ is a closed form
and $\aaa(u,v) = \aaa(u{\mathds{1}}_M,v{\mathds{1}}_M) +\aaa(u{\mathds{1}}_{M^c},v{\mathds{1}}_{M^c}) $
for all $u,v$.
If there is no such decomposition of $\mathfrak{h}$, the latter form is called
\emph{irreducible}. Note that reducibility can be rephrased in terms of the
semigroup and the resolvent:

\begin{theorem}
 Let $\mathfrak{h}$ be as above. Then the following conditions are equivalent:
\begin{itemize}
 \item $\mathfrak{h}$ is irreducible.
  \item $T_t$ is positivity improving, for every $t>0$, i.e. $f\ge 0$ and $f\not=0$ implies that $T_tf>0$ a.e.
 \item $(H+E)^{-1}$ is positivity improving for every $E<\inf\sigma(H)$.
\end{itemize}
\end{theorem}

In \cite{LenzSV-pre} we will show that for a strictly local Dirichlet form $\aaa$ as above and a measure perturbation $\nu\in\ \mathcal{M}_{R,0}-\mathcal{M}_{R,1}$, irreducibility of $\aaa$ implies irreducibility of $\aaa+\nu$.

\section{Positive weak solutions and the associated transformation}
Throughout this section we consider a strongly local, regular Dirichlet form,
$(\aaa,\DD)$ on $X$ and denote by $\Gamma: \DD_{\text{loc}}\times\DD_{\text{loc}}\to
\C{M}(X)$ the associated energy measure. We will be concerned with weak
solutions $\Phi$ of the equation
\begin{equation}\label{eq21}
(H_0+V)\Phi=E\cdot \Phi,
\end{equation}
where $H_0$ is the operator associated with $\aaa$ and $V$ is a realvalued,
locally integrable potential. In fact, we will consider a somewhat more
general framework, allowing for measures instead of functions, as presented in
the previous section. Moreover, we stress the fact that (\ref{eq21}) is formal
in the sense that $\Phi$ is not assumed to be in the operator domain of
neither $H_0$ nor $V$. Here are the details.

\begin{definition}
Let $U\subset X$ be open and $\nu\in\C{M}_{R,0}(U)$ be a signed Radon measure on $U$ that charges no set of capacity zero.
Let  $E\in\RR$ and $\Phi\in L^2_{\text{loc}}(U)$. We say that $\Phi$ is a
\emph{weak supersolution} of $(H_0+\nu)\Phi=E\cdot\Phi$ in $U$ if:
\begin{itemize}
\item[(i)] $\Phi\in\DD_{\text{loc}}(U)$,
\item[(ii)] $\tilde{\Phi}d\nu\in\C{M}_R(U)$,
\item[(iii)] $\forall \varphi\in\DD\cap C_c(U), \varphi\ge 0:$
$$
\aaa[\Phi,\varphi]+\int_U\varphi \tilde{\Phi}d\nu\ge E \cdot (\Phi |\varphi) .
$$
\end{itemize}
We call $\Phi$ a
\emph{weak solution} of $(H_0+\nu)\Phi=E\cdot\Phi$ in $U$ if equality holds in (iii) above
(which extends to all $\varphi\in\DD\cap C_c(U)$).
If $V\in L^1_{\text{loc}}(U)$ we say that $\Phi$ is a
\emph{weak (super-)solution} of $(H_0+V)\Phi=E\cdot\Phi$ in $U$ if it is a weak (super-)
solution of  $(H_0+\nu)\Phi=E\cdot\Phi$ for $\nu=Vdm$.
\end{definition}
\begin{remark}
\begin{itemize}
\item[(1)] If $\nu=Vdm$ and $V\in L^2_{\text{loc}}(U)$, then property (ii) of the Definition above is satisfied.
\item[(2)] If $\Phi\in L^\infty_{\text{loc}}(U)$ and $\nu\in\C{M}_R(U)$ then (ii) of the Definition above is satisfied.
\item[(3)] If $\nu\in\C{M}_R(U)$ satisfies (ii) above then $\nu - Edm\in\C{M}_R(U)$
satisfies (ii) as well and any weak solution of
$(H_0+\nu)\Phi=E\cdot\Phi$ in $U$ is a weak solution of
$(H_0+\nu-Edm)\Phi=0$ in $U$. Thus it suffices to consider the case $E=0$.
\item[(4)] If $\Phi$ is a weak solution on $U$, then
$$
\aaa[\Phi,\varphi]+\int_U\varphi \tilde{\Phi}d\nu =  E \cdot (\Phi |\varphi) .
$$
for all $\varphi \in\mathcal{D}\cap L^\infty_c (U)$. This follows easily from (b) of the approximation Lemma \ref{l:approximation}. (Note that we can indeed approximate within $U$ by first choosing an appropriate $\eta$ with compact support in $U$ according to Lemma \ref{cutoff}.)
\end{itemize}
\end{remark}

We will deal with functions $\Phi \in \DD_{\text{loc}}$ with  $\Phi >0$. If $\Phi$ is such a function  and $\Phi^{-1} \in L^\infty_{\text{loc}}$, we can use the chain rule and suitable smoothed version of the function $x\mapsto 1/x$ to conclude that $\Phi^{-1}$ must belong to $\DD_{\text{loc}}$ as well. This will be used various times in the sequel.

\smallskip

Here comes   the first half of the Allegretto-Piepenbrink Theorem in a
general form.

\begin{theorem}\label{t:transformation}
Let $(\aaa,\DD)$ be a regular, strictly local Dirichlet form, $H_0$ be the associated operator and $\nu\in\C{M}_{R,0}(U)$. Suppose
that $\Phi$  is a weak solution  of $(H_0+\nu)\Phi=E\cdot\Phi$ in $U$ with $\Phi>0$ $m$-a.e. and $\Phi, \Phi^{-1}\in L^\infty_{\text{loc}}(U)$.
Then, for all $\varphi, \psi\in \DD\cap L^\infty_{c}(U)$:
\[
\aaa[\varphi,\psi]+\nu[\varphi,\psi] = \int_U\Phi^2d\Gamma(\varphi\Phi^{-1},
\psi\Phi^{-1}) + E\cdot (\varphi |\psi) .
\]
In particular, $\aaa+\nu\ge E$ if furthermore  $U=X$.
\end{theorem}
\begin{proof}
The ``in particular'' is clear as the desired inequality holds on $\DD\cap C_c (X)$ and the form is regular by Lemma \ref{l:approximation}.

For the rest of the proof we may assume $E=0$ without restriction, in view of
the preceding remark.  Without loss of generality we may also assume that $\varphi$ and $\psi$ are real valued functions.  We now  evaluate the RHS of the above equation, using the following
identity. The Leibniz rule implies that  for arbitrary $w\in \DD_{\text{loc}}(U)$:
$$
0=d\Gamma(w,1)=d\Gamma(w,\Phi\Phi^{-1})=\Phi^{-1}d\Gamma(w,\Phi)+\Phi d\Gamma(w,\Phi^{-1})\qquad( \bigstar)
$$
Therefore, for $\varphi, \psi\in \DD\cap C_c(X)$:
\begin{eqnarray*}
\int_X\Phi^2d\Gamma(\varphi\Phi^{-1},
\psi\Phi^{-1}) &=& \int_X\Phi d\Gamma(\varphi,\psi\Phi^{-1})+
\int_X\Phi^2\varphi d\Gamma(\Phi^{-1},\psi\Phi^{-1})\\
(\mbox{by symmetry})\;\:&=& \int_X d\Gamma(\varphi,\psi)+\int_X\Phi \psi d\Gamma(\varphi,\Phi^{-1})\\
& & + \int_X\Phi^2\varphi d\Gamma(\psi\Phi^{-1},\Phi^{-1})\\
&=& \aaa[\varphi,\psi] + \int_X\Phi^2 d\Gamma(\varphi\psi\Phi^{-1},\Phi^{-1})\\
(\mbox{ by $(\bigstar)$})\;\:&=& \aaa[\varphi,\psi]- \int_X d\Gamma(\varphi\psi\Phi^{-1},\Phi)\\
 &=& \aaa[\varphi,\psi]-\aaa[\varphi\psi\Phi^{-1},\Phi].
 \end{eqnarray*}
 As $\Phi$ is a weak solution we can now use part (4) of the previous remark to continue the computation by
 \begin{eqnarray*}
...&=& \aaa[\varphi,\psi]-\left( -\nu[\varphi\psi\Phi^{-1},\Phi] \right)\\
 &=& \aaa[\varphi,\psi]+\nu[\varphi,\psi] .
\end{eqnarray*}
This finishes the proof.
\end{proof}

We note a number of consequences of the preceding theorem. The first is rather
a consequence of the proof, however:

\begin{corollary}\label{c2.4}
Assume that there is a weak supersolution $\Phi$ of $(H_0+\nu)\Phi=E\cdot\Phi$ on $X$ with
$\Phi>0$ $m$-a.e. and $\Phi,\Phi^{-1}\in L^\infty_{\text{loc}} (X)$. Then
$\aaa +\nu \ge E$.
\end{corollary}
For the \emph{Proof} we can use the same calculation as in the proof of the Theorem
with $\varphi=\psi$ and use the inequality instead of the equality at the end.
\begin{remark}
\begin{itemize}
\item[(1)] We can allow for complex measures $\nu$ without problems. In the context of PT--symmetric
  operators there is recent interest in this type of Schr\"{o}dinger
  operators, see \cite{BenderBM-99}
\item[(2)] Instead of measures also certain distributions could be included. Cf \cite{HerbstS-78} for such singular perturbations.
\end{itemize}
\end{remark}

We will extend Theorem  \ref{t:transformation} to all of $\varphi,\psi\in \cD$. This is somewhat technical. The main part is done in the next three propositions. We will assume the situation (S):

\begin{itemize}
\item[(S)]
Let $(\aaa,\DD)$ be a regular, strictly local Dirichlet form, $H_0$ be the associated operator and $\nu \in \mathcal{M}_{R,0}- \mathcal{M}_{R,1}$.  Suppose
that $\Phi$  is a weak solution  of $(H_0+\nu)\Phi=E\cdot\Phi$ in $X$ with $\Phi>0$ $m$-a.e. and $\Phi, \Phi^{-1}\in L^\infty_{\text{loc}}(X)$.
\end{itemize}

\begin{proposition}\label{convergence} Assume (S).
Let $u\in \DD(\cE+\nu)$ be given. Let $(u_n)$ be a sequence in $\DD (\cE + \nu)\cap L^\infty_c (X)$ which converges to $u$ with respect to $\|\cdot\|_{\cE + \nu}$. Then,
  $\varphi u_n \Phi^{-1}$ and $\varphi u \Phi^{-1}$ belong to  $\DD (\cE + \nu) $ and
  $$ \|\varphi u_n \Phi^{-1} - \varphi u \Phi^{-1}\|_{\cE+ \nu} \to 0, n\to \infty$$
  for any $\varphi \in \DD\cap C_c (X)$ with $d\Gamma (\varphi)\leq C  d m$ for some $C>0$. In particular,  $u\Phi^{-1}$ belongs to $\DD_{\text{loc}}$.
\end{proposition}
\begin{proof} Without loss of generality we assume $E=0$.

As shown in Lemma \ref{localversion} $\varphi \Phi^{-1}$ 
belongs to $\DD\cap L^\infty_c$. Hence, $ \varphi u_n \Phi^{-1} = u_n (\varphi \Phi^{-1})$ is a product of functions in $\DD\cap L^\infty_c $ and therefore belongs to $\DD\cap L^\infty_c$ as well.

As $\varphi \Phi^{-1}$ belongs to $L^\infty$, the sequence
$\varphi u_n \Phi^{-1}$ converges to $\varphi u \Phi^{-1}$ in $L^2 (X,m)$. It therefore suffices to show that $\varphi u_n \Phi^{-1}$ is a Cauchy sequence with respect to $\|\cdot\|_{\cE + \nu}$.

As $(u_n)$ is Cauchy with respect to $\|\cdot\|_{\cE + \nu}$ and $\varphi \Phi^{-1}$  is bounded, convergence of the $\nu$ part is taken care of and it
 suffices to show that
$$\cE (\varphi (u_n - u_m) \Phi^{-1}) \to 0, n,m\to \infty.$$
Let $K$ be the compact support of $\varphi$.
 Let $c>0$ be an upper bound for $\Phi^{-2}$ on $K$. Choose $n,m\in \N$ and set $v:= u_n - u_m$. Then, we can calculate
 \begin{eqnarray*}
 \cE (\varphi v \Phi^{-1}) &=&\int_K d\Gamma (\varphi v \phim)\\
 &=& \int_K \frac{1}{\Phi^2} \Phi^2  d\Gamma (\varphi v \phim)\\
 &\leq & c \int_K \Phi^2  d\Gamma (\varphi v \phim)\\
(\mbox{Previous theorem})\;\: &=& c (\cE (\varphi v) + \nu (\varphi v))\\
 &=& c (\cE( \varphi (u_n - u_m)) + \nu (\varphi (u_n - u_m))).
  \end{eqnarray*}
 Now, convergence of $\nu (\varphi (u_n - u_m))$ to $0$ for $n,m\to\infty$ can easily be seen (with arguments as at  the beginning of the proof). As for $\cE( \varphi (u_n - u_m))$ we can use Leibniz rule  and Cauchy-Schwarz and $d\Gamma (\varphi) \leq C \, dm$ to compute
 \begin{eqnarray*} \cE (\varphi (u_n - u_m)) &=& \int_K d\Gamma (\varphi (u_n - u_m))\\
 &=& \int \varphi^2 d\Gamma (u_n - u_m) + 2 \int \varphi (u_n - u_m) d\Gamma (\varphi,u_n - u_m)\\
  & &+ \int |u_n - u_m|^2 d\Gamma (\varphi)\\
 &\leq &
 2 ( \int \varphi^2 d\Gamma (u_n - u_m) + \int |u_n - u_m|^2 d\Gamma (\varphi))\\
 &\leq & 2 \|\varphi\|^2 \cE (u_n - u_m) + 2 C \int |u_n - u_m|^2 dm.
 \end{eqnarray*}
This gives  easily the desired convergence to zero and  $(\varphi u_n \phim)$ is a Cauchy sequence with respect to $\|\cdot\|_{\cE + \nu}$.

We now turn to a proof of the last statement: By Lemma \ref{cutoff}, for any compact $K$ we can find a $\varphi$ satisfying the assumptions of the proposition with $\varphi \equiv 1$ on $K$. Then,  $\varphi u \phim$ belongs to $\DD$ by  the above argument and agrees with $u \phim$ on $K$ be construction.
\end{proof}

\begin{proposition} Assume (S).
Let $u\in \DD(\cE+\nu)$ be given. Let $(u_n)$ be a sequence in $\DD (\cE + \nu)\cap L^\infty_c (X)$ which converges to $u$ with respect to $\|\cdot\|_{\cE + \nu}$.  Then,   $$\int \psi d\Gamma ( u_n \phim) \to \int \psi d\Gamma ( u \phim)$$
for any $\psi \in L^\infty_c (X)$.
\end{proposition}
\begin{proof} We start with an intermediate claim.

\smallskip

 Claim. For any $\psi \in L^\infty (X)$ and $\varphi \in \DD \cap C_c (X)$ with $d\Gamma (\varphi) \leq C \, dm$ for some $C>0$, we have  $\int \psi d\Gamma (\varphi u_n \phim) \to \int \psi d\Gamma (\varphi u \phim)$.

\smallskip

Proof of the claim.  By triangle inequality, the difference between the terms in question can be estimated by
$$ | \int \psi d\Gamma (\varphi (u- u_n) \phim, \varphi u_n \phim)| + |\int \psi d\Gamma (\varphi  u\phim,\varphi (u- u_n)\phim)|.$$
By Cauchy Schwarz inequality these terms can be estimated by
$$\|\psi\|_\infty \cE ((\varphi (u_n -u) \phim)^{1/2} \cE (\varphi u_n \phim)^{1/2}$$
and
$$ \|\psi\|_\infty \cE ((\varphi (u_n -u) \phim)^{1/2} \cE (\varphi u \phim)^{1/2}.$$
The previous proposition gives that $\cE ( \varphi (u_n -u) \phim)\to 0$, $n\to \infty$ and the claim follows.

\bigskip

Let now $\psi \in L^\infty_c (X)$ be given.  Let $K$ be the compact support of $\psi$. We use Lemma \ref{cutoff} to find
$\varphi\in C_c (X)\cap \DD$ with $\varphi \equiv 1$ on $K$ and $d\Gamma (\varphi)\leq C\, dm$. for some $C>0$. Locality gives
$$ \int \psi d\Gamma ( u_n \phim)  = \int \psi d\Gamma (\varphi  u_n \phim) $$
and
$$ \int \psi d\Gamma ( u \phim) = \int \psi d\Gamma ( \varphi u \phim)$$
and the proposition   follows from the claim.
\end{proof}

\begin{proposition}\label{stern} Assume (S).
Let $u\in \DD(\cE+\nu)$ be given. Let $(u_n)$ be a sequence in $\DD (\cE + \nu)\cap L^\infty_c (X)$ which converges to $u$ with respect to $\|\cdot\|_{\cE + \nu}$.  
Then,   $$\int \Phi^2 d\Gamma ( u_n \phim) \to \int \Phi^2 d\Gamma ( u \phim).$$
\end{proposition}
\begin{proof} Without loss of generality we assume $E=0$. We start with the following claim.

\smallskip

Claim. $\cE (u) + \nu (u) \geq \int \Phi^2 d\Gamma (u \phim)$.

Proof of claim. By convergence of $u_n$ to $u$ w.r.t. $\|\cdot\|_{\cE+ \nu}$  and the last theorem, we have
 $$\cE (u) + \nu (u) = \lim_{n\to \infty} \cE (u_n) + \nu (u_n)=\lim_{n\to \infty} \int \Phi^2  d\Gamma (u_n \phim).$$
Let $\chi_R$ be the characteristic function of a the ball with radius $R$ around a fixed point in $X$ and $\psi= \Phi^2 \chi_R	$.
With this choice of $\psi$ the preceeding proposition can be applied.
 Now, the claim follows easily from  a Fatou type argument when $R$ tends to infinity.

\medskip

We now note that for fixed $n\in \NN$, the sequence $(u_m - u_n)_m$ converges to $u-u_n$ w.r.t. $\|\cdot\|_{\cE+ \nu}$. We can therefore  apply the claim to $u - u_n$ instead of $u$. This gives

$$ \cE (u-u_n) + \nu (u- u_n) \geq \int \Phi^2 d\Gamma ((u -u_n) \phim) \geq 0$$
for any $n\in \NN$.  As the left hand side converges to zero for $n\to \infty$, so does the right hand side.

Mimicking the argument given in the proof of the Claim  of the previous proposition, we can now conclude the desired statement.
\end{proof}

\begin{corollary}\label{c:transformation}
Let $(\aaa,\DD)$ be a regular, strictly local Dirichlet form, $H_0$ be the associated operator and $\nu \in \mathcal{M}_{R,0}- \mathcal{M}_{R,1}$.  Suppose
that $\Phi$  is a weak solution  of $(H_0+\nu)\Phi=E\cdot\Phi$ in $X$ with $\Phi>0$ $m$-a.e. and $\Phi, \Phi^{-1}\in L^\infty_{\text{loc}}(X)$. Then, for   all $\varphi,\psi \in D(\cE + \nu)$, the products
 $\varphi \phim,\psi \phim$ belong to $\DD_{\text{loc}}$ and
 the formula
\begin{equation}\label{eq:formula}
\aaa[\varphi,\psi]+\nu[\varphi,\psi] = \int_X\Phi^2d\Gamma(\varphi\Phi^{-1},
\psi\Phi^{-1}) + E\cdot (\varphi |\psi)
\end{equation}
holds.
\end{corollary}
\begin{proof} Without loss of generality we assume $E=0$.  It suffices to consider $\varphi = \psi$.
By Proposition \ref{convergence}, $\varphi \phim$ belongs to $\DD_{\text{loc}}$. According to Lemma \ref{l:approximation}, we can  choose  a sequence $(\varphi_n)$ in $\DD\cap L^\infty_c (X)$ converging to $\varphi$ w.r.t. $\|\cdot\|_{\cE+ \nu}$. This convergence and the last theorem then give
$$\cE(\varphi) + \nu (\varphi) = \lim_{n\to \infty} \cE(\varphi_n) + \nu(\varphi_n)= \lim_{n\to \infty} \int \Phi^2  d\Gamma (\varphi_n \phim).$$
The previous proposition then yields the desired formula.
\end{proof}

\section{The existence of positive weak solutions below the spectrum}
As noted in the preceding section, we find that $H_0+\nu\ge E$ whenever
$\aaa+\nu$ is closable and admits a positive weak solution of
$(H_0+\nu)\Phi=E\Phi$. In this section we prove the converse under suitable
conditions. We use an idea from \cite{Simon-82,CyconFKS-87} where the corresponding
statement for ordinary Schr\"{o}dinger operators on $\RR^d$ can be found. A
key property is related to the celebrated \emph{Harnack inequality}.
\begin{definition}{\rm
\begin{itemize}
\item [(1)] We say that $H_0+\nu$ \emph{satisfies a Harnack inequality} for
  $E\in\RR$ if, for every relatively compact, connected open $X_0\subset X$ there is a
  constant $C$ such that all positive weak solutions $\Phi$ of
  $(H_0+\nu)\Phi=E\Phi$ on $X_0$ are locally bounded and satisfy
$$
\mbox{esssup}_{B(x,r)}u\le C  \mbox{essinf}_{B(x,r)}u ,
$$
for every $B(x,r)\subset X_0$ where esssup and essinf denote the essential supremum and infimum.
\item [(2)] We say that $H_0+\nu$ satisfies the \emph{Harnack principle} for
  $E\in\RR$ if for every relatively compact, connected open subset $U$ of $X$
  and every sequence $(\Phi_n)_{n\in\NN}$ of nonnegative solutions of
  $(H_0+\nu)\Phi=E\cdot\Phi$ in $U$ the following implication holds: If, for
  some measurable subset $A \subset U$ of positive measure
$$
\sup_{n\in\NN}\|\Phi_n{\mathds{1}}_A\|_2 <\infty
$$
then, for all compact $K\subset U$ also
$$
\sup_{n\in\NN}\| \Phi_n{\mathds{1}}_K\|_2 <\infty .
$$
\item [(3)]
 We say that $H_0+\nu$ satisfies the \emph{uniform Harnack principle}
 if for every bounded intervall $I\subset\RR$,
 every relatively compact, connected open subset $U$ of $X$ and every sequence $(\Phi_n)_{n\in\NN}$ of nonnegative solutions of $(H_0+\nu)\Phi=E_n\cdot\Phi$ in $U$ with $E_n\in I$ the following implication holds:
If, for some measurable subset $A \subset U$ of positive measure
$$
\sup_{n\in\NN}\|\Phi_n{\mathds{1}}_A\|_2 <\infty
$$
then, for all compact $K\subset U$ also
$$
\sup_{n\in\NN}\| \Phi_n{\mathds{1}}_K\|_2 <\infty .
$$
\end{itemize}
}
\end{definition}
Note that validity of a Harnack principle implies that a nonnegative
weak solution $\varPhi$ must vanish identically if it vanishes on a
set of positive measure (as $\varPhi_n := n \varPhi$ has vanishing
$L^2$ norm on the set of positive measure in question).  Note also
that validity of an Harnack inequality extends from balls to compact
sets by a standard chain of balls argument.  This easily shows that
$H_0+\nu$ satisfies the Harnack principle for $E\in\RR$ if it obeys a
Harnack inequality for $E\in\RR$.  Therefore, many situations are
known in which the Harnack principle is satisfied:

\begin{remark}\label{r:Hansen}
\begin{itemize}
\item[(1)] For $\nu\equiv 0$ and $E=0$ a Harnack inequality holds, whenever $\aaa$ satisfies a Poincar\'e and a volume doubling property; cf \cite{BiroliM-95} and the discussion there.
\item[(2)] The most general results for $H_0=-\Delta$ in terms of the measures
  $\nu$ that are allowed seem to be found in \cite{Hansen-99a}, which also
  contains a thorough discussion of the literature prior to 1999. A crucial
  condition concerning the measures involved is the Kato condition and the
  uniformity of the estimates from \cite{Hansen-99a} immediately gives that the
  uniform Harnack principle is satisfied in that context.  Of the enormous
  list of papers on Harnack's inequality, let us mention
  \cite{AizenmanS-82,Biroli-01, BiroliM-06, ChiarenzaFG-86,Hansen-99a, HoffmannHN-95, Kassmann-07,
    Moser-61,Saloff-Coste-95, Serrin-64,Sturm-94c,Sturm-96}
\end{itemize}
\end{remark}

Apart from the Harnack principle there is a second property that will be
important in the proof of existence of positive general eigensolutions at
energies below the spectrum:
We say that $\aaa$ satisfies the \emph{local compactness property} if
$D_0(U):=\overline{D\cap C_c(U)}^{\|\cdot\|_\aaa}$ is compactly
embedded in $L^2(X)$ for every relatively compact open $U\subset X$.
(In case of the classical Dirichlet form this follows from Rellich's
Theorem on compactness of the embedding of Sobolev spaces in $L^2$.)

\begin{theorem}\label{thm3.3}
  Let $(\aaa,\DD)$ be a regular, strictly local, irreducible Dirichlet
  form, $H_0$ be the associated operator and
  $\nu\in\C{M}_{R,0}-\C{M}_{R,1}$.  Suppose that $\aaa$ satisfies the
  local compactness property and $X$ is noncompact. Then, if
  $E<\inf\sigma(H_0+\nu)$ and $H_0+\nu$ satisfies the Harnack
  principle for $E$, there is an a.e. positive solution of
  $(H_0+\nu)\Phi=E\Phi$.
\end{theorem}
\begin{proof}
  Let $E<\inf\sigma(H_0+\nu)$.  Since $X$ is noncompact, locally
  compact and $\sigma$-compact, it can be written as a countable union
$$
X=\bigcup_{R\in\NN}U_R,\; U_R\mbox{  open, relatively compact  }, \overline{U_R}
\subset U_{R+1} ;
$$
where the $U_R$ can be chosen connected, as $X$ is connected, see \cite{LenzSV-pre} for details.

For $n\in\NN$ let $g_n\in L^2(X)$ with $\supp g_n\subset X\setminus U_{n+2}$, $g_n\ge 0$ and $g_n\not= 0$. It follows that
$$
\Phi_n:= (H_0+\nu+E)^{-1}g_n\ge 0
$$
is nonzero and is a weak solution of $(H_0+\nu)\Phi=E\Phi$ on $X\setminus
\supp g_n$, in particular on the connected open subset $U_{n+2}$. Since
$(H_0+\nu+E)^{-1}$ is positivity improving, it follows that $
\|\Phi_n{\mathds{1}}_{U_1}\|_2 >0$. By multiplying with a positive constant we
may and will assume that $ \|\Phi_n {\mathds{1}}_{U_1}\|_2 =1$ for all
$n\in\NN$. We want to construct a suitably convergent subsequence of
$(\Phi_n)_{n\in\NN}$ so that the corresponding limit $\Phi$ is a positive weak
solution.

First note that by the Harnack principle, for fixed $R\in\NN$ we
know that
$$
\sup_{n\ge R} \|\Phi_n {\mathds{1}}_{U_R}\|_2 <\infty ,
$$
since all the corresponding $\Phi_n$ are nonnegative solutions on
$U_{R+2}$.  In particular, $(\Phi_n {\mathds{1}}_{U_R})_{n\in\NN}$ is bounded
in $L^2(X)$ and so has a weakly convergent subsequence. By a standard diagonal
argument, we find a subsequence, again denoted by $(\Phi_n)_{n\in\NN}$, so
that $\Phi_n {\mathds{1}}_{U_R}\to \Psi_R$ weakly in $L^2(X)$ for all
$R\in\NN$ and suitable $\Psi_R$. As multiplication with ${\mathds{1}}_{U_R}$
is continuous and hence also weak-weak continuous, there is $\Phi\in
L^2_{\text{loc}}(X)$ such that $\Psi_R=\Phi {\mathds{1}}_{U_R}$. We will now perform
some bootstrapping to show that the convergence is, in fact, much better than
just local weak convergence in $L^2$ which will imply that $\Phi$ is the
desired weak solution.

Since for fixed $R>0$ and $n\ge R$ the $\Phi_n$ are nonnegative solutions on $U_{R+2}$ the Caccioppoli inequality, cf \cite{BoutetdeMonvelLS-08} implies that
$$
\int_{U_R}d\Gamma(\Phi_n)\le C  \int_{U_{R+1}}\Phi_n^2dm
$$
is uniformly bounded w.r.t. $n\in\NN$. Combined with Leibniz rule and
Cauchy Schwarz inequality this directly gives that $\int_{U_R}d\Gamma(\psi
\Phi_n)$ is uniformly bounded w.r.t. $n\in \NN$ for every $\psi\in \DD$ with
$d\Gamma (\psi) \leq dm$ (see \cite{BoutetdeMonvelLS-08} as well). Therefore, by Lemma \ref{cutoff}, we can find
for suitable cut-off functions $\eta_R\in \DD\cap C_c(X)$ with
${\mathds{1}}_{U_R}\le\eta_R\le {\mathds{1}}_{U_{R+1}}$ such that  the sequence
$(\eta_R\Phi_n)$ is bounded in $(D, \|\cdot\|_\aaa)$.

The local compactness property implies that $ (\eta_R\Phi_n)$ has an
$L^2$-convergent subsequence. Using a diagonal argument again, we see that
there is a common subsequence, again denoted by $(\Phi_n)_{n\in\NN}$, such
that
$$
\Phi_n {\mathds{1}}_{U_R}\to \Phi {\mathds{1}}_{U_R}\mbox{  in  }L^2(X)\mbox{ as  }n\to\infty
$$
for all $R\in\NN$.

As a first important consequence we note that $\Phi\not=0$, since $\|\Phi {\mathds{1}}_{U_1}\|_2 =\lim_n\|\Phi_n {\mathds{1}}_{U_1} \|_2 =1$.

Another appeal to the Caccioppoli inequality gives that
$$
\int_{U_R}d\Gamma(\Phi_n-\Phi_k)\le C  \int_{U_{R+1}}(\Phi_n-\Phi_k)^2dm\to 0\mbox{  as  }n,k\to\infty .
$$
Therefore, by the same reasoning as above, for every $R\in\NN$ the sequence
$(\eta_R\Phi_n)$ converges in $(\DD, \|\cdot\|_\aaa)$. Since this convergence
is stronger than weak convergence in $L^2(X)$, its limit must be $\eta_R\Phi$,
so that the latter is in $\DD$. We have thus proven that $\Phi\in
\DD_{\text{loc}}(X)$. Moreover, we also find that
$$
\aaa[\Phi_n,\varphi]\to \aaa[\Phi,\varphi]\mbox{  for all  }\varphi\in \DD\cap C_c(X) ,
$$
(since, by strong locality, for every  cut-off function $\eta\in \DD\cap C_c(X)$ that is $1$ on $\supp\varphi$, we get
$$
\aaa[\Phi_n,\varphi]=\aaa[\eta\Phi_n,\varphi]\to
\aaa[\eta\Phi,\varphi]=\aaa[\Phi,\varphi].) $$

We will now deduce convergence of the potential term.  This will be done in
two steps. In the first step we infer convergence of the $\nu_-$ part from
convergence w.r.t. $\|\cdot\|_\aaa$ and the relative boundedness   of
$\nu_-$. In the second step, we use the fact that $\Phi$ is a weak solution
to reduce   convergence of the $\nu_+$ part to convergence
w.r.t. $\|\cdot\|_\aaa$ and convergence of the $\nu_-$ part. Here are the
details:

Consider cut-off functions $\eta_R$ for $R\in\NN$ as above. Due to convergence
in $(\DD, \|\cdot\|_\aaa)$, we know that there is a subsequence of
$(\eta_R\Phi_n)_{n\in\NN}$ that converges q.e., see \cite{Fukushima-80} and
the discussion in Section 1. One diagonal argument more will give a
subsequence, again denoted by $(\Phi_n)_{n\in\NN}$, such that the
$\tilde{\Phi}_n$ converge to $\tilde{\Phi}$ q.e., where $ \tilde{}$ denotes
the quasi-continuous representatives. Since $\nu$ is absolutely continuous
w.r.t capacity we now know that the $\tilde{\Phi}_n$ converge to
$\tilde{\Phi}$ $\nu$-a.e. Moreover, again due to convergence in $(\DD,
\|\cdot\|_\aaa)$, we know that $(\eta_R\tilde{\Phi_n})_{n\in\NN}$ is
convergent in $L^2(\nu_-)$ as $\nu_-\in\C{M}_{R,1}$. Its limit must coincide
with $\eta_R\tilde{\Phi}$, showing that $\tilde{\Phi}d\nu_-\in\C{M}_{R}$.

We now want to show the analogous convergence for $\nu_+$; we do so by
approximation and omit the $\tilde{}$ for notational simplicity.  By simple
cut-off procedures, every $\varphi\in \DD_c (X) \cap L^\infty (X)$ can be
approximated w.r.t. $\|\cdot\|_\aaa$ by a uniformly bounded sequence of
continuous functions in $\DD$ with common compact support.  Thus, the equation
$$
\aaa[\Phi,\varphi]+\nu[\Phi,\varphi]=E \cdot(\Phi|\varphi) ,
$$
initially valid for $\varphi\in \DD\cap C_c(X)$ extends to $\varphi\in \DD_c
(X) \cap L^\infty (X)$ by continuity.
 Therefore, for arbitrary $k\in\NN$, and $R< \min (n-2, m-2)$
\begin{small}\begin{eqnarray*}
 \int_{|\Phi_n-\Phi_m|\le k}(\Phi_n-\Phi_m)^2\eta_Rd\nu_+ &\le&
\int(\Phi_n-\Phi_m)\{(-k)\vee(\Phi_n-\Phi_m)\wedge k\}\eta_Rd\nu_+ \\
&=&
\nu_+[(\Phi_n-\Phi_m),\{(-k)\vee(\Phi_n-\Phi_m)\wedge k\}\eta_R]\\
&=&
E((\Phi_n-\Phi_m)|\{\ldots\}\eta_R)+\nu_-[(\Phi_n-\Phi_m),\{\ldots\}\eta_R)]\\
& & - \aaa[(\Phi_n-\Phi_m),\{(-k)\vee(\Phi_n-\Phi_m)\wedge k\}\eta_R]
\end{eqnarray*}               \end{small}
By what we already know about convergence in $\DD$, $L^2$ and $L^2(\nu_-)$, the RHS goes to zero as $n,m\to\infty$, independently of $k$.  This gives the desired convergence of $\eta_R\tilde{\Phi}_n$, the limit being $\eta_R\tilde{\Phi}$ since this is the limit pointwise.

Finally, an appeal to the Harnack principle gives that $\Phi$ is positive a.e. on every $U_R$ and, therefore, a.e. on X.
\end{proof}
\begin{remark}\label{cpt}
  That we have to assume that $X$ is noncompact can easily be seen by
  looking at the Laplacian on a compact manifold. In that situation
  any positive weak solution must in fact be in $L^2$ due to the
  Harnack principle. Thus the corresponding energy must lie in the
  spectrum. In fact, the corresponding energy must be the infimum of the spectrum as we will show in the next theorem.   The theorem  is standard. We include a proof for completeness reasons.
\end{remark}

\begin{theorem}
  Let $(\aaa,\DD)$ be a regular, strictly local, irreducible Dirichlet
  form, $H_0$ be the associated operator and
  $\nu\in\C{M}_{R,0}-\C{M}_{R,1}$.  Suppose that $X$ is compact and  $\aaa$ satisfies the
  local compactness property. Then, $H_0 + \nu$ has compact resolvent. In particular,   there exists a positive weak solution
  to $(H_0 + \nu) \Phi = E_0 \Phi$ for $E_0:=\inf \sigma (H_0 +\nu)$.  This solution is unique (up to a factor) and  belongs to $L^2
  (X)$. If $H_0 + \nu$ satisfies a Harnack principle, then  $E_0$ is the only value in $\RR$ allowing for a positive  weak solution.
\end{theorem}
\begin{proof} As $X$ is compact, the local compactness property gives that the operator associated to $\aaa$ has compact resolvent.
In particular, the sequence of eigenvalues of $H_0$ is given by the minmax principle and tends to $\infty$. As $\nu_+$ is a nonnegative operator   and $\nu_-$ is form bounded with bound less than one, we can apply the minmax principle to $H_0 + \nu$ as well to obtain empty essential spectrum.

In particular, the infimum of the spectrum is an eigenvalue.
By irreducibility and abstract principles, see e.g.
\cite{ReedS-78}, XIII.12, the corresponding  eigenvector must have constant sign and if a Harnack principle holds then  any other energy
allowing for a positive weak solution must  be an eigenvalue as
well (as discussed in the previous remark). As there can not be two different eigenvalues with positive solutions, there can not be another energy with a positive weak solution.
\end{proof}

Combining the results for the  compact and noncompact case we get:

\begin{corollary}\label{c3.5}
 Let $(\aaa,\DD)$ be a regular, strictly local, irreducible Dirichlet form, $H_0$ be the associated operator and $\nu\in\C{M}_{R,0}-\C{M}_{R,1}$.
  Suppose that $\aaa$ satisfies the local compactness property and  $H_0+\nu$ satisfies the Harnack principle for all $E\in\RR$. Then,
$$
\inf \sigma (H_0+\nu)\le \sup\{ E\in\RR| \exists \mbox{ a.e. positive weak solution }(H_0+\nu)\Phi=E\Phi\} .
$$
\end{corollary}
This doesn't settle the existence of a positive weak solution for the groundstate energy $\inf \sigma (H_0+\nu)$ in the noncompact case. The uniform Harnack principle settles this question:

\begin{theorem} Let $(\aaa,\DD)$ be a regular, strictly local, irreducible Dirichlet form, $H_0$ be the associated operator,  $\nu\in\C{M}_{R,0}-\C{M}_{R,1}$.
  Suppose that $\aaa$ satisfies the local compactness property and $H_0+\nu$ satisfies the uniform Harnack principle. Then there is an a.e. positive weak
solution of $(H_0+\nu)\Phi=E\Phi$ for  $E= \inf \sigma (H_0+\nu)$.
\end{theorem}
\begin{proof} It suffices to consider the case of noncompact $X$.
 Take a sequence $(E_n)$ increasing to $E=\inf \sigma (H_0+\nu)$. From Theorem \ref{thm3.3} we know that there is an a.e. positive solution $\Psi_n$ of
$(H_0+\nu)\Phi=E_n\Phi$. We use the exhaustion $(U_R)_{R\in\NN}$ from the proof of Theorem \ref{thm3.3} and assume that
$$
\|\Psi_n {\mathds{1}}_{U_1} \|_2 =1\mbox{  for all  }n\in\NN .
$$
As in the proof of Theorem \ref{thm3.3} we can now show that we can
pass to a subsequence such that $(\eta_R\Psi_n )$ converges in $\DD$,
$L^2(m)$ and $L^2(\nu_+ + \nu_-)$ for every $R\in\NN$. The crucial point is
that the uniform Harnack principle gives us a control on $\|
\eta_R\Psi_n \|_2$, uniformly in $n$, due to the norming condition
above. With aruments analogous to those in the proof of Theorem
\ref{thm3.3}, the assertion follows.
\end{proof}
Note that Corollaries \ref{c2.4} and \ref{c3.5} together almost give
$$
\inf \sigma (H_0+\nu) = \sup\{ E\in\RR| \exists \mbox{ a.e. positive weak solution }(H_0+\nu)\Phi=E\Phi\} .
$$
The only problem is that for the ``$\ge$'' from Corollary \ref{c2.4} we
would have to replace a.e. positive by a.e. positive and
$\Phi, \Phi^{-1}\in L^\infty_{\text{loc}}$. This, however, is fulfilled whenever  a Harnack inequality holds.

\begin{corollary}\label{c3.7}
 Let $(\aaa,\DD)$ be a regular, strictly local, irreducible Dirichlet form, $H_0$ be the associated operator and $\nu\in\C{M}_{R,0}-\C{M}_{R,1}$.
  Suppose that $\aaa$ satisfies the local compactness property and  $H_0+\nu$ satisfies a Harnack inequality for all $E\in\RR$. Then,
$$
\inf \sigma (H_0+\nu) = \sup\{ E\in\RR| \exists \mbox{ a.e. positive weak solution }(H_0+\nu)\Phi=E\Phi\} .
$$
\end{corollary}

\section{Examples and applications}
We discuss several different  types of operators to which our results can be applied.
Parts of the implications have been known before.
However, previous proofs dealt with
each of the mentioned operators separately, while we have a uniform argument of proof.

\subsection*{Examples}
Classical examples of operators for which our results have been known before
can be found in \cite{Allegretto-74, Allegretto-79, Allegretto-81,MossP-78,Piepenbrink-74,Piepenbrink-77,CyconFKS-87}.
They concern Schr\"odinger operators and, more generally, symmetric elliptic second
order differential operators on unbounded domains in $\RR^d$, whose coefficients satisfy certain regularity conditions.
For Laplace-Beltrami operators on Riemannian manifolds
the Allegretto-Piepenbrink theorem  has been stablished in \cite{Sullivan-87b}.
\medskip

Here we want to concentrate on two classes of examples which have attracted attention more recently:
Hamiltonians with singular interactions and quantum graphs.

\subsubsection*{Hamiltonians with singular interactions}
These are operators acting on $\RR^d$ which may be formally written as $H = -\Delta - \alpha \delta(\cdot - M)$
where $\alpha$ is a positive real and $M \subset \RR^d$ is a manifold of codimension one satisfying certain regularity conditions, see e.g.\,\cite{BrascheEKS-94}
or Appendix K of \cite{AlbeverioGHH-04}.  In fact, the delta interaction can be given a rigorous interpretation as a measure $\nu_M$ concentrated on
the manifold  $M$. More precisely, for any Borel set $B\subset \RR^d$,  one sets $\nu_M(B) :=\vol_{d-1}(B \cap M)$
where $\vol_{d-1}$ denotes the $(d-1)$-dimensional Hausdorff measure on $M$.
In \cite{BrascheEKS-94}, page 132, one can find suitable regularity conditions on $M$
under which the measure $\nu_\gamma$ belongs to  the class $\mathcal{M}_{R,1}$.
Thus the singular interaction operator $H$ falls into our general framework, cf.~Remark  \ref{r:Hansen}.

If $M$ is a $C^2$-regular, compact curve  in $\RR^2$ the
essential spectrum of $H$ equals $\sigma_{ess}(-\Delta) = [0,\infty)$, cf.~\cite{BrascheEKS-94}.
On the other hand, the bottom of the spectrum of $H$ is negative and consists consequently of an eigenvalue.
This can be seen using the proof of Corollary 11 in \cite{Brasche-03}.
In Section 3 of  \cite{Exner-05} it has been established
that the ground state is nondegenerate and the corresponding eigenfunction strictly positive.
This corresponds to part of our Theorem 3.3.

\subsubsection*{Quantum graphs}
Quantum graphs are given in terms of a metric graph $X$ and a Laplace (or more generally) Schr\"odinger operator $H$
defined on the edges of $X$ together with a set of (generalised) boundary conditions
at the vertices which make $H$ selfadjoint.
To make sure that we are dealing with a strongly-local Dirichlet form
we restrict ourselves here to the case of so called free or Kirchoff boundary conditions.
A function in the domain of the corresponding quantum graph Laplacian $H_0$ is continuous at each vertex
and  the boundary values of the derivatives obtained by approaching the vertex along incident edges
sum up to zero. Note that any non-negative Borel measure on $X$  belongs to the class $\mathcal{M}_{R,0}(X)$.
For $ \nu_+ \in \mathcal{M}_{R,0}(X)$  and $ \nu_-\in \mathcal{M}_{R,1}(X)$ the
quantum graph operator $H= H_0 + \nu_+-\nu_-$ falls into our framework.

See Section 5 of \cite{BoutetdeMonvelLS-08}  for a more detailed discussion of the
relation between Dirichlet forms and quantum graphs.

\subsection*{Applications}

The ground state transformation which featured in Theorem \ref{t:transformation} and Corollary \ref{c:transformation}
can be used to obtain a formula for the lowest spectral gap. To be more precise
let us assume that $\mathcal{E}$, $\nu$ and $\Phi$ satisfy the conditions of Theorem \ref{t:transformation}
with $U=X$.
Assume in addition that $\Phi$ is in $\mathcal{D} (\cE+ \nu)$.
Then $\Phi$ is an eigenfunction of $H$ corresponding to the eigenvalue $E = \min \sigma (H)$.
We denote by
\[
 E':= \inf \{ \aaa[u,u] + \nu[u,u] \mid u \in \DD, \|u \|=1, u \perp \Phi \} \,
\]
the second lowest eigenvalue below the essential spectrum of $H$,
or, if it does not exist, the bottom of $\sigma_{ess}(H)$.
Then we obtain the following formula
\begin{equation}\label{e:gap-formula}
 E' -E = \inf_{ \{ u \in \DD (\cE+\nu) , \|u \|=1, u \perp \Phi \}} \int_X \Phi^2 d\Gamma(u \Phi^{-1},u \Phi^{-1}) \
\end{equation}
which determines the lowest spectral gap. It has been used in \cite{KirschS-87,KondejV-06a,Vogt}
to derive lower bounds on the distance between the two lowest eigenvalues of different classes of Schr\"odinger operators
(see \cite{SingerWYY-85} for a related approach). In \cite{KirschS-87} bounded potentials are considered, in
\cite{KondejV-06a} singular interactions along curves in $\RR^2$ are studied, and \cite{Vogt}
generalises these results using a unified approach based on Kato-class measures.

If for a subset  $U \subset X$ of positive measure and a function $u \in \{ u \in \DD, \|u \|=1, u \perp \Phi \}$
the non-negative measure $\Gamma(u\Phi^{-1},u\Phi^{-1})$ is absolutely continuous with respect to $m$, one can exploit formula
\eqref{e:gap-formula} to derive the following estimate
(cf.{} Section~3 in \cite{Vogt}, and \cite{KirschS-87,KondejV-06a} for similar bounds).
Denote by $\gamma(u\Phi^{-1})= \frac{d\Gamma(u\Phi^{-1},u\Phi^{-1})}{dm}$ the Radon-Nykodim derivative. Then
\[
\int_U \Phi^2 d\Gamma(u \Phi^{-1},u \Phi^{-1})
\ge \frac{1}{m(U)} \inf_U \Phi^2
  \left( \int_U \sqrt{\gamma(u\Phi^{-1})}dm\right)^2 \,
\]
In specific situations one can chose $u$ to be an eigenfunction associated to the second eigenvalue $E'$ and use
geometric properties of $\Phi$ and $u$ to derive explicit lower bounds on the spectral gap.
\bigskip

Other uses of the ground state transformation include the study of $L^p$-$L^q$ mapping properties of the semigroup associated to $\aaa$
\cite{DaviesS-84}
and the proof of Lifschitz tails \cite{Mezincescu-87}.\\[5mm]

\textbf{Acknowledgements.}
We would like to thank J.~Brasche and N.~Peyerimhoff for helpful discussions
concerning the previous literature.

\def\cprime{$'$}\def\polhk#1{\setbox0=\hbox{#1}{\ooalign{\hidewidth
  \lower1.5ex\hbox{`}\hidewidth\crcr\unhbox0}}}

\Addresses
\end{document}